\documentclass[11pt]{amsart}
\usepackage{fullpage}
\usepackage{mydefs}
\usepackage{amssymb}
\usepackage{graphicx}

\begin{document}
\title[Capacity Maximization in Wireless Networks]{On a Game Theoretic Approach to Capacity Maximization in
Wireless Networks}

\author[E. \'{A}sgeirsson]{Eyj\'{o}lfur Ingi \'{A}sgeirsson}
\address[E. \'{A}sgeirsson]{School of Science and Engineering\\
Reykjavik University\\
Reykjavik 101, Iceland}
\email{eyjo@hr.is}

\author[P. Mitra]{Pradipta Mitra}
\address[P. Mitra]{School of Computer Science\\
Reykjavik University\\
Reykjavik 101, Iceland}
\email{ppmitra@gmail.com}

\begin{abstract}
We consider the capacity problem (or, the single slot scheduling problem) in wireless networks. 
Our goal is to maximize the number of successful
connections in arbitrary wireless networks where a transmission is
successful only if the signal-to-interference-plus-noise ratio at the
receiver is greater than some threshold. We study a game theoretic
approach towards capacity maximization introduced
by Andrews and Dinitz (INFOCOM 2009) and Dinitz (INFOCOM 2010). 
We
prove vastly improved bounds for the game theoretic algorithm. 
In doing so, we achieve the first \emph{distributed} constant factor approximation
algorithm for capacity maximization for the uniform power assignment. 
When compared to the optimum where links may use an arbitrary power
assignment, we prove a $O(\log \Delta)$ approximation, where $\Delta$ is the ratio 
between the largest and the smallest
link in the network. 
This is an exponential improvement of the approximation factor compared
to existing results for distributed algorithms. All our results work for links located in \emph{any} metric space.
In addition, we provide simulation studies
clarifying the picture on distributed algorithms for capacity
maximization.
\end{abstract}

\maketitle

\section{Introduction}
The question of maximizing the \emph{capacity} of
a wireless communication network is a well-studied problem. The setting in which we study
this is the following: there are a set of \emph{links}, where each link
represents a potential transmission from a sender to a receiver. The \emph{capacity}
then is the maximum number of links that can successfully transmit at once, ie, in the same
time \emph{slot}.

A central question in this context is how to model interference between various attempted
transmissions in the network. The models used in the literature
can essentially be divided into two types. First, there is the \emph{protocol model} where
interference is modeled by a \emph{interference graph}, and a
transmission is successful if and only if none of the neighbors
of the transmission in this graph also choose to transmit at the same time. 
Thus capacity maximization becomes equivalent to the maximum independent set problem.
With appropriate restrictions placed on the structure of the graph, a number of solutions
have been proposed to this problem and its variants (eg, \cite{Schneider:2008:LDM:1400751.1400758,Nieberg:2008:ASW:1383369.1383380,Erlebach:2001:PAS:365411.365562}).

It is well-known though, that graph-based protocols are not very good in capturing reality and this has been 
demonstrated both theoretically and
experimentally~\cite{MaheshwariJD08,Moscibroda2006Protocol}. As a result,
a lot of recent algorithmic work has focused on the so-called \emph{physical model} or the SINR model.
We will describe the precise model in Section \ref{sec:model}, but we give an outline here.
Basically, in the SINR model every communication link has a power at which it transmits
(this power may be pre-determined, say due to hardware limitations, or the sender may be able to 
choose its own transmission power).
The power received from such a transmission across the space fades away from
its source in a physically reasonable way.
Given this,
every attempted communication interferes with every other attempted communication,
but in differing quantities depending on the distances between the links.
 Though this model is still an abstraction of reality, it is believed
to model wireless networks better than the \emph{protocol model}, and we will focus solely on
this model in this work.

The capacity of random networks in the SINR model was studied in
the highly-cited work by
Gupta and Kumar \cite{kumar00}, and a large number of papers have pursued the same theme. 
Algorithmic results on worst case instances, on the other hand, have only recently garnered attention,
starting with the important work of Moscibroda and Wattenhofer \cite{MoWa06}.
Since then, a large body of work has been produced for this problem (see 
\cite{gouss2007,Moscibroda2006Protocol,BE06a,HW09,us:esa09full,FKV09,KV10} and many references therein).

Even within the framework of the SINR model, a number of variations exist.
First, there is the question of the space in which the communication links
lie. Assuming that the space is a 2-dimensional Euclidean plane is 
natural, yet clearly a simplification.
Obstructions, air density, geometry of the enclosing space,
antenna directionality, and terrains complicate the picture, bringing the Euclidean assumption
 into question (indeed, the space may not even be a metric). 
 Most results cited above focus on the Euclidean plane, or a generalization thereof
 known as \emph{fading metrics} \cite{us:esa09full}.
 Naturally, we would like to provide results with as much generality as possible,
 and as we will see, our results are applicable to completely general metrics.
A related issue is that  of the \emph{path loss exponent} $\alpha$,
which defines how the signal fades away from its source.
Most approximation results have used the 
assumption that $\alpha > 2$, which is used in a crucial fashion along with the Euclidean plane assumption.
Though this assumption has some justification, it
is known that $\alpha$ can actually be equal to or even smaller
than 2 in real networks (see \cite{SWTFA95}). Our results will simply assume that $\alpha > 0$.

Also, as hinted before, the transmission powers used by links 
may be either \emph{fixed} or \emph{arbitrary}.
In the fixed version, each link must use a predetermined power, whereas in the 
arbitrary case, the algorithm may select different powers to increase capacity. The most
important fixed power assignment for our purposes is the uniform power assignment,
where each link must use the same power. In fact, our algorithm will use
uniform power. Given this, we can compare the capacity achieved by our algorithm to either
the maximum capacity achievable using uniform power, or the 
with the maximum capacity achievable when links can choose a arbitrary power assignment to increase capacity. 
We shall do both.

Finally, there is the crucial issue of centralized vs. distributed algorithms. Results for approximation algorithms
for the SINR model have been almost exclusively focused on centralized algorithms, including
most of the work cited before. While centralized
algorithms are interesting in their own right, and also may be useful in practice, clearly 
in a variety of real scenarios, distributed algorithms would be preferable. 
We are aware of only a few works that address this important issue in a rigorous manner. First, recently
\cite{KV10} proposed a $O(\log^2 n)$ approximation algorithm for the \emph{scheduling}
problem (where one is trying to find the minimum number of slots required to successfully
transmit all communication requests). This is related, yet not directly comparable to the 
capacity maximization problem.

On the capacity maximization problem, two recent related papers by 
Andrews and Dinitz \cite{DBLP:conf/infocom/AndrewsD09}
and by Dinitz \cite{Dinitz2010} tackle the question in a distributed setting. Since
we adopt their approach and those papers provide the main point of comparison for this work, 
its is worth discussing the main ideas in them.
In \cite{DBLP:conf/infocom/AndrewsD09} the authors model the distributed
setting as a game where each link is a player and the power settings are the 
pure strategies. Then they show that at a mixed Nash equilibrium, the expected
number of transmissions is within a factor of $O(\Delta^{2 \alpha})$ of the
optimum using arbitrary power assignments (where $\Delta$ is the ratio between
the largest and smallest link in the network). However, this is a non-algorithmic
result. In \cite{Dinitz2010}, Dinitz uses the concept of a \emph{no-regret algorithm} to convert this
structural result into an algorithmic one. The author shows that if each link uses a \emph{no-regret algorithm},
then after a certain number of rounds, the network reaches a position similar to that of a mixed Nash equilibrium, and thus
acheives a $O(\Delta^{2 \alpha})$ approximation. 
As mentioned in those papers, the approach is very robust and surprisingly versatile. It can handle malicious links, gracefully handle
new links joining the network, and allows individual links to use different algorithms as long as each uses
a \emph{no-regret} algorithm (a number of which exist, all of which are fairly simple).
Yet, from an approximation stand-point, the result
is not very promising, specially when compared to $O(1)$-approximation for uniform power \cite{HW09}, and 
$O(\log \log \Delta \cdot \log n)$ for arbitrary power assignments \cite{us:esa09full} (these results use centralized
algorithms, though). 
In addition, the result works only for euclidean plane with $\alpha > 2$ (and fading metrics
with appropriately large $\alpha$).
Our goal is to improve these results and generalize them to general metrics while only requiring  $\alpha > 0$.
An appealing aspect of this proposed task is that the machinery of
\emph{no-regret algorithm}s can be used essentially as a black-box, thus improvements to the structural
result provide, after routine modifications, improved algorithms. This is what we achieve. Specifically,
we show that after convergence, the algorithm achieves $O(1)$-approximation factor compared to the optimum
using uniform power and $O(\log \Delta)$-approximation factor compared to the optimum using
arbitrary power assignments. For uniform power, we thus essentially get best possible bounds, while
for arbitrary power, we improve the bounds exponentially. Surprisingly, our proof approach is quite different
and arguably much simpler compared to that of \cite{Dinitz2010}.

\section{Communication Model and Results}
\label{sec:model}


We assume that there is a set $L$ of links, where
each link $v \in L$ represents a potential transmission from a sender
$s_v$ to a receiver $r_v$, each points in a metric space. 
The distance between two points $x$ and $y$ is denoted $d(x,y)$.

The distance from
$v$'s sender to $w$'s receiver is denoted $d_{vw} = d(s_v, r_w)$.
The length of link $v$ is denoted 
by $\ell_v = d(s_v, r_v)$.

The set may be associated with a \emph{power assignment}, which is an assignment of a transmission
power $P_v$ to be used by each link $v \in L$.
We assume $0 \leq P_v \leq P_{\max}$ for some 
fixed $P_{\max}$, for all $v$. 
The setting $P_v = 0$ means the sender is not transmitting. For simplicity we will assume $P_{\max} = 1$
without loss of generality.
The signal received at point
$y$ from a sender at point $x$ with power $P$  is $P/d(x, y)^\alpha$ where the constant 
$\alpha > 0$ is the
\emph{path-loss exponent}. 

We can now describe the  \emph{physical} or SINR-model of interference. In this model, a receiver $r_v$
successfully receives a message from the sender $s_v$ if and only if the
following condition holds:
\begin{equation}
 \frac{P_v/\ell_v^\alpha}{\sum_{\ell_w \in S \setminus  \{\ell_v\}}
   P_w/d_{wv}^\alpha + N} \ge \beta, 
 \label{eq:sinr}
\end{equation}
where $N$ is the environmental noise, the constant $\beta \ge 1$ denotes the minimum
SINR (signal-to-interference-noise-ratio) required for a message to be successfully received,
and $S$ is the set of concurrently scheduled links in the same \emph{slot}.

We say that $S$ is \emph{SINR-feasible} (or simply \emph{feasible}) if (\ref{eq:sinr}) is
satisfied for each link in $S$. 
Let $\Delta = \frac{l_{\max}}{l_{\min}}$ where $l_{\max}$ and $l_{\min}$ are respectively, the maximum and minimum lengths in $L$.

\begin{defn}
The affectance $a^P_w(v)$ of link $\ell_v$ caused by another link $\ell_w$,
with a given power assignment $P$,
is the interference of $\ell_w$ on $\ell_v$ relative to the power
received, or
  \[ a^P_{w}(v) 
     = \min\left\{1, c_v \frac{P_w}{P_v} \cdot \left(\frac{\ell_v}{d_{wv}}\right)^\alpha\right\},
  \] 
where $c_v = \beta/(1 - \beta N \ell_v^\alpha/P_v)$. 
\end{defn}

The definition of affectance was introduced in \cite{GHWW09} and achieved the form
we are using in \cite{KV10}.
When referring to uniform power (where $P_v = 1$ for all $v$)
we drop the superscript $P$. 
Also, let $a^P_v(v) = 0$.
Using the idea of affectance, Eqn. \ref{eq:sinr} can be rewritten as $\sum_{u \in S}a^P_u(v) \leq 1$ for all $v \in S$.

We will use the notation $OPT$ to denote the largest set that is feasible using uniform power, and $OPT_P$
to denote the largest set that is feasible using some arbitrary power assignment. Thus the maximum
capacity of the network is $|OPT|$ or $|OPT_P|$, depending the flexibility one allows on the power assignments.

As has become common in a distributed setting \cite{DBLP:conf/infocom/AndrewsD09,Dinitz2010,KV10} we assume
\begin{equation}
\beta \geq c_v \geq \frac{\beta}{2}
\label{cvlimit}
\end{equation}
which essentially means that received signal at every receiver is somewhat larger than what is needed to succeed given
the environmental noise.

A last piece of terminology we need is a \emph{$\delta$-signal} set. A \emph{$\delta$-signal} set is a set of links where the
affectance on any link is at most $1/\delta$.
A set is feasible iff it is a 1-signal set.

\subsection{A note about the model}
Our model is different from the model used in \cite{DBLP:conf/infocom/AndrewsD09,Dinitz2010} in one
small detail. In those works the signal received at $y$ from a sender at point $x$ with power $P$  is $\min\{1, P/d(x, y)^\alpha\}$ instead of $P/d(x, y)^\alpha$ in our case. Let us call the former the ``bounded" model, and call ours
the ``unbounded" model.
Both models have been used in the literature, our choice is motivated by it being somewhat more elegant, and
some of the machinery we use from previous works being developed in the unbounded model. Though the two
models are not equivalent, from a technical standpoint, 
the difference is rather minor. All results we use and prove are easily transferrable to the bounded
model by routine modification, in the form of separate handling of lengths less than $1$.

The point that needs to be clarified is how to compare the results for the two models.
The reader will notice that the bounds in \cite{DBLP:conf/infocom/AndrewsD09,Dinitz2010}
are expressed in terms of $l_{\max}$ (they claim a $O(l_{\max}^{2\alpha})$ approximation), where as ours is expressed in terms of 
 $\Delta = \frac{l_{\max}}{l_{\min}}$ (for example, $O(\log \Delta)$ approximation for arbitrary powers). 
In fact, for purposes of this comparison, these two entities are the same. 
If converted to the unbounded model, the results in \cite{DBLP:conf/infocom/AndrewsD09,Dinitz2010} will have to replace $l_{\max}$ with $\Delta$ (which is mentioned in \cite{DBLP:conf/infocom/AndrewsD09} as well), 
whereas in the bounded
model our dependence on $\Delta$ changes to dependence on $\frac{l_{\max}}{\max\{1, l_{\min}\}} \leq l_{\max}$. Thus for purpose of comparison between the two sets of results,
$\Delta$ and $l_{\max}$ are interchangeable, and we will simple use $\Delta$ in all cases.

\subsection{Results}
\subsubsection{Notion of capacity in a distributed setting}
Our goal is maximize the capacity of a wireless network, where capacity is defined as the number
of links that can be simultaneously scheduled in a single slot.
This definition concerning a  \emph{single} slot requires us to carefully consider what we mean
by capacity in a distributed setting. As we will see, our distributed algorithms involve multiple
rounds (ie, slots) where the same set of links are trying to transmit, resulting in a \emph{average}
large capacity after a while. There are a number of ways we can view this. First, we can consider
this to be a distributed capacity \emph{determination} algorithm. That is to say, an algorithm to find or
approximate the network capacity, for example to use it as a benchmark for performance evaluation.
This is perhaps the most theoretically satisfying explanation.

On the other hand we may wish to see the algorithm as one deployed to \emph{achieve} high capacity. If the rounds taken for the algorithm to converge is small, we may absorb the number of rounds needed
it into the approximation factor. Our results show that, theoretically speaking, the number of rounds taken is quite high.
On the other hand, simulations indicate that in practice it may not be as bad. 
Finally, if the communication over links are sustained (ie,
they require a large number of slots to complete), the notion of average capacity makes sense as well.

We can now state our results. We will prove the following.

\begin{theorem}
There exists a class of distributed algorithms such that if
 every sender uses an algorithm from this class, then after $O((\frac{n}{|OPT|})^2 \log n)$ rounds the
average number of successful connections is $\Omega(|OPT|)  = \Omega(|OPT_P|/\log \Delta)$ with probability
at least $1 - \frac{1}{n}$. Also, all algorithms in the class use uniform power, that is, each sender $s_v$
either transmits at full power $P_v = 1$, or does not transmit at all.
\label{thm:mainresult}
\end{theorem}
This can be compared to the result in \cite{Dinitz2010}, where a $O(\Delta^{2 \alpha})$ approximation
factor is proven for the arbitrary power case. For uniform power, nothing better than $O(\Delta^{2 \alpha})$ 
is claimed in that work. Also, the results in \cite{Dinitz2010} work for the plane with $\alpha > 2$ (and a generalization
of euclidean metrics known as fading metrics with appropriately large $\alpha$). In comparison, our result is
stated for completely general metrics and for any $\alpha > 0$.

The following corollary is useful to state explicitly.
\begin{corollary}
There exists a randomized distributed  $O(1)$-approximation algorithm to determine, with high probability, 
the capacity 
of a wireless network under uniform power. For arbitrary power assignments, the same algorithm achieves
a  $O(\log \Delta)$ approximation, with high probability.
\end{corollary}
What is remarkable here is that a such results for general metrics have not been published
even for centralized algorithms. The best known result for uniform power is $O(1)$-approximation for fading metrics
with $\alpha$ larger than the doubling dimension \cite{HW09}. For arbitrary power, the best
known algorithm is the $O(\log \log \Delta \cdot \log n)$ algorithm due to Halld\'{o}rsson \cite{us:esa09full},
again for fading metrics. Both of these algorithms are centralized. We are aware of a recent unpublished
result \cite{HM10} that has achieved $O(1)$-approximation for uniform power for general metrics (the algorithm there is
centralized as well).

\section{Game theoretic basics}
As discussed before, the basic game theoretic approach considered in this paper 
was developed by \cite{DBLP:conf/infocom/AndrewsD09}
 and the first algorithmic results based on that was derived in \cite{Dinitz2010}. In this section, we review
 and collect important concepts and results from those two papers. The reader may find more information
 in \cite{DBLP:conf/infocom/AndrewsD09,Dinitz2010}.

The games we are
interested in have $n$ players in which every player
has exactly two possible actions. Let $\mathcal{A} = \{0, 1\}^n$ be the space
of all possible actions (or possible ``strategy profiles") for the game, i.e. given a point
$A \in \mathcal{A}$, the $i^{th}$ coordinate $a_i$ represents the action used by
player $i$ in profile $A$. For each player $i$ there is a utility function $\alpha_i:\mathcal{A} \rightarrow \mathbb{R}$ denoting
how good certain actions for that player are. We will
sometimes want to consider modifications of strategy profiles:
given $A\in \mathcal{A}$, let $A \oplus a'_i$
be the strategy set obtained by player
$i$ changing its action from $a_i$ to $a'_i$. We will use superscripts to
denote time, so $A^t$ will be the action set at time $t$ and $a^t_i$ will
be the action taken by player $i$ at time $t$.
The following definition is crucial:

\begin{defn}
The regret of player $i$ at time $T$ given strategy
profiles $A^1, A^2, \ldots A^T$ is
$$\max_{a_i \in \{0,1\}} \frac{1}{T} \sum^T_{t = 1} \alpha_i (A^t \oplus a_i) - \frac{1}{T} \sum^T_{t = 1} \alpha_i (A^t)$$
\end{defn}
Having low regret essentially means that the player has done almost as
well on average as the best single action would have done.
We refer the reader to \cite{Dinitz2010} for a detailed discussion
of the ideas and historical context related to this notion. What is 
directly relevant though is the following powerful result, asserting the existence of a \emph{no-regret} algorithm:

\begin{theorem}[\cite{DBLP:journals/siamcomp/AuerCFS02}]
There is an algorithm that has regret at
most $O(\sqrt{\frac{\log(T/\delta)}{T}})$ with probability at least $1 - \delta$ for any
$\delta> 0$, for any game with a constant number of possible actions
per player.
\label{schapire}
\end{theorem}
This result is applicable for the \emph{bandit} model, where the player only knows the utility that it gained as a result
of taking an action,
not what would have happened if it played another action. This definition matches the situation in wireless
networks, where we assume that the sender knows if the transmission succeeded, but does not know anything
if it did not try to transmit at all. Also, since the algorithm is applicable to a single player, it is by definition ``distributed"
(ie, the algorithm solely depends on the utility the individual player gains in the course of the game). The moral
here is that with this tool at our disposal, to achieve a approximation algorithm, all we need is a structural result
comparing the optimum in one hand, and the average number of successful connections when each link has no regret
(ie, small regret) on the other.

A specific algorithm meeting the claims of Thm. \ref{schapire} is provided in \cite{DBLP:journals/siamcomp/AuerCFS02}. 
A similar guarantee was given for the Randomized Weighted
Majority Algorithm by Littlestone and Warmuth \cite{Littlestone:1994:WMA:184036.184040}. We will use this latter
algorithm in our simulations. These algorithms are all surprisingly simple.
 
 Now we can define the game:
Each sender $s_v$ is a
player, with two possible strategies: transmit at power $P_v = 1$ (full power) 
or don't transmit at all (ie, power $P_v = 0$). 
A transmitter has utility $1$ if
the transmission succeeds. It has utility -1 if it attempts to transmit but fails,
and utility 0 if it does not transmit at all.

Note that the game (and thus the algorithm) uses uniform power. So we can ask
two questions. First, how well does the algorithm do when compared to the optimum 
using uniform power? Second, how well does it do when compared to the optimum 
where links can use some arbitrary power assignment (within the upper bound of $1$)? We will
provide answers to both questions.
 
Let $T$ be some time at which all transmitters have regret
at most $\epsilon$. We seek to prove that the average number of
successful connections per slot up to time $T$ has been close to $|OPT|$.
Let $q_u$ be the fraction of times at which $s_u$
chose to transmit, and let $x_u$ be the fraction
of times at which $u$ transmitted successfully. Then $Q = \sum_u q_u$
 is the average number of attempted transmissions and $X = \sum_u x_u$
 is the average number of successful transmissions, so
we are trying to prove that $|X|$ is close to $|OPT|$. The following
lemma relates $|Q|$ and $|X|$.

\begin{lemma}[\cite{Dinitz2010}]
$X \leq Q \leq 2X + \epsilon n$
\end{lemma}

Now if we can prove that $Q = \Omega(OPT)$ and then choose $\epsilon = \frac{|OPT|}{c_1 n}$
for a suitably large constant $c_1$ we can then assert $X = \Omega(OPT)$.
Setting
 $\delta=1/n^2$ and $T \geq (\frac{c_1 n}{OPT})^2 \log n$, by Thm. \ref{schapire}, we will achieve the
 required bound on $\epsilon$ for a single sender
with probability at least $1-1/n^2$ , and thus with probability $1 - \frac{1}{n}$
\emph{every} sender will have regret at most $\frac{OPT}{c_1 n}$.
The $O(1)$-approximation claimed in Thm. \ref{thm:mainresult} now follows.
All that is required is the bound $Q = \Omega(OPT)$ which we will prove in the next section.

The following important observation is embedded in \cite{Dinitz2010}, but it is useful to collect it in one lemma.
\begin{lemma}
Let $G = \{u: q_u \geq \frac{1}{2} - \epsilon\}$. Define $f_u$ as the fraction of time the link $u$ would have failed
if it tried to transmit,
irrespective of whether or not it did actually try. If the links have regret no more than $\epsilon$, then for all $u \in L\setminus G$, $f_u \geq \frac{1}{4}$.
\label{lem:pbad}
\end{lemma}
\begin{proof}
For contradiction, assume $f_u < \frac{1}{4}$. Then setting $q_u = 1$ would give an expected utility $> \frac{3}{4} - \frac{1}{4} > \frac{1}{2}$.
However, since $u \in L \setminus G$, $q_u < \frac{1}{2} - \epsilon$ thus its utility is also less than $\frac{1}{2} - \epsilon$.
Thus regret for $u$  is $> \frac{1}{2} - \frac{1}{2} + \epsilon = \epsilon$ which is a contradiction of the fact that $u$ has regret
at most $\epsilon$.
\end{proof}

\section{Derivation of Results}
First we need a basic result about feasible sets.
\begin{lemma}
Let $L$ be a feasible set. Define the set $L' = \{u \in L: \sum_{v \in L} a_u(v) \leq 2\}$.
Then, $|L'| \geq |L|/2$.
 \label{markov1}
\end{lemma}
\begin{proof}
Since $L$ is feasible, we know for all $v \in L$, $\sum_{u \in L}a_u(v) \leq 1$. Thus,
\begin{eqnarray}
\sum_{v\in L}\sum_{u \in L}a_u(v) \leq |L|
\label{setsum}
\end{eqnarray}
 If the claim of the Lemma is false then $|L'| < |L|/2$, thus $|L \setminus L' | > |L|/2$.
 Now,
 \begin{eqnarray*}
&& \sum_{v\in L}\sum_{u \in L}a_u(v) = \sum_{u\in L}\sum_{v \in L}a_u(v)  \\
& \geq & \sum_{u\in L \setminus L'}\sum_{v \in L}a_u(v) > |L|/2 \cdot 2  = |L|
 \end{eqnarray*}
which contradicts Eqn. \ref{setsum}.
The first inequality
 follows from the fact that $L \setminus L' \subseteq L$ and $a_u(v) \geq 0$. 
 The second inequality is a consequence 
 of the fact that  $\sum_{v \in L} a_u(v) > 2$ for all $u \in L \setminus L'$ and that $|L \setminus L' | > |L|/2$.
\end{proof}

Now we state the main technical Lemma.
\begin{lemma}
Suppose at time $T$ each sender has regret at most $\epsilon$ (where $\epsilon$ is very small). Then at time $T$,
$$Q = \Omega(OPT)$$
where $Q = \sum_{u \in L} q_u$ as defined before, and $OPT$ is the optimum capacity for uniform
power.
\label{lem:maintech}
\end{lemma}
\begin{proof}
As in Lemma \ref{lem:pbad} let $G = \{u \in OPT: q_u \geq \frac{1}{2} - \epsilon\}$ and let $OPT' = OPT \setminus G$.
If $|G| > |OPT|/2$, then $Q \geq \sum_{u \in G} q_u \geq (\frac{1}{2} - \epsilon) |G| = \Omega(|OPT|)$ and we would
be done. So let us assume that $|OPT'| \geq |OPT|/2$.

Now, let $OPT''  = \{u \in OPT': \sum_{v \in OPT'} a_u(v) \leq 2\}$. By Lemma \ref{markov1}, $|OPT''| \geq |OPT'|/2$,
and therefore, $|OPT''| \geq |OPT|/4$.

By Lemma \ref{lem:pbad}, for all $v \in OPT''$, $f_v \geq \frac{1}{4}$. Recall that $f_v$ is the fraction of time
that a transmission from $s_v$ fails. Defining $a(v)$ to be the total affectance on $v$ from other links, we can say
\begin{equation}
f_v \equiv \Pro(a(v) > 1) \geq \frac{1}{4}
\label{fvb1}
\end{equation}
Now average affectance on 
$v$ is $\Ex(a(v)) = \sum_{u \in L} q_u a_u(v)$.
 By Markov's inequality,  
 \begin{equation}
\Pro(a(v) \geq 5\Ex(a(v)) ) \leq \frac{1}{5}
\label{exavb}
 \end{equation}
Comparing Eqns. \ref{fvb1} and \ref{exavb}, we can easily see that $5\Ex(a(v)) \geq 1$, or

\begin{equation}
\sum_{u \in L} q_u a_u(v) \geq \frac{1}{5}; \forall v \in OPT''
\end{equation}
Note that the sum is over all links in $L$.

Summing all the inequalities together, we get
\begin{eqnarray}
\sum_{v \in OPT''}\sum_{u \in L} q_u a_u(v) \geq \frac{|OPT''|}{5} \nonumber \\
\Rightarrow \sum_{u \in L} q_u \sum_{v \in OPT''}  a_u(v) \geq \frac{|OPT''|}{5} \label{eqn:mainsum}
\end{eqnarray}

We now claim,  $\sum_{v \in OPT''}  a_u(v) \leq c$ for all $u \in L$ and some constant $c$. 
We will prove this claim in Lemma \ref{o1}. But let us
see how this leads to proof of the present Lemma. 
We get from
Eqn. \ref{eqn:mainsum} that,
\begin{eqnarray*}
&&c \sum_{u \in L} q_u  \geq \frac{|OPT''|}{5} \\
& \Rightarrow & \sum_{u \in L} q_u \geq \frac{|OPT''|}{5c} = \Omega(|OPT|)
\end{eqnarray*}
The last equality comes from recalling that $|OPT''| \geq |OPT|/4$.
\end{proof}

Now we prove the promised Lemma. First, we need a known result.
\begin{lemma}[\cite{us:esa09full}]
Let $\ell_u, \ell_v$ be links in a $q^\alpha$-signal set under any power assignment.
Then, $d_{uv} \cdot d_{vu} \ge q^2 \cdot \ell_u \ell_v$. 
\label{lem:ind-separation}
\end{lemma}

\begin{lemma}
\label{o1}
Assume $R$ is a feasible set under uniform power such that for all $z \in R$, $\sum_{v \in R} a_z(v) \leq 2$. Then for any other link $u$, 
$\sum_{v \in R}  a_u(v) = O(1)$.
\end{lemma}
\begin{proof}
We use the signal strengthening technique by Halld\'{o}rsson and Wattenhofer \cite{HW09}. That is, we decompose 
the set $R$ to $\lceil 2 \cdot 3^\alpha/\beta \rceil$ sets, each a $3^\alpha$-signal set.
We prove the claim for one such set, since there are only constantly many such sets, the overall claim holds.
Let us reuse the notation $R$ to be such a $3^\alpha$ signal set.

Consider a link $w \in R$ such that $d(s_w, s_u)$ is minimum, ie, for all $v \in R$, $d(s_v, s_u) \geq d(s_w, s_u)$.
Also consider the link $w' \in R$ such that $d(r_{w'}, s_u)$ is minimum, ie, for all $v \in R$, $d(r_{v}, s_u) \geq d(r_{w'}, s_u)$.

Let, $D = d( s_w, s_u)$. Now we claim that for all links, $z \in R, z \neq w'$
\begin{equation}
d(s_u, r_z) \geq \frac{1}{2} D
\label{eqn:dist1}
\end{equation}
For contradiction, assume $d(s_u, r_z) < \frac{1}{2} D$. Then,  $d(r_{w'}, s_u) < \frac{1}{2} D$, by definition.
Now, again by the definition of $w$, $d(s_z, s_u) \geq D$ and $d(s_{w'}, s_u) \geq D$. Thus $l_z > \frac{D}{2}$
and $l_{w'} > \frac{D}{2}$. On the other hand $d(r_z, r_{w'}) < \frac{D}{2} + \frac{D}{2} < D$.
Now, $d_{w'z} \cdot d_{zw'} \leq (l_{w'} + d(r_z, r_{w'}))(l_{z} + d(r_z, r_{w'})) <  (l_{w'} + D)(l_{z} + D)
< 9 l_{w'} l_z$, which contradicts Lemma \ref{lem:ind-separation}.

Now, $d_{wz}  = d(s_w, r_z) \leq d(s_w, s_u) + d(s_u, r_z) \leq 3 d(s_u, r_z) = 3 d_{uz}$, where the last inequality
follows from Eqn. \ref{eqn:dist1}. Then, for all $z \in R$, $z \neq w'$, 
$a_u(z) = c_z \left(\frac{l_z}{d_{uz}}\right)^\alpha \leq c_z 3^{\alpha}\left(\frac{l_z}{d_{wz}}\right)^\alpha = 3^{\alpha} a_w(z)$.
Finally,
\begin{eqnarray*}
&& \sum_{v \in R}  a_u(v) = a_u(w') + \sum_{v \in R \setminus \{w'\}}  a_u(v)\\
&\leq &1 + \sum_{v \in R \setminus \{w'\}}  a_u(v) \leq 1 + 3^\alpha \sum_{v \in R \setminus \{wÕ\}}  a_z(v) \\
&\leq &1 + 3^\alpha \cdot 2 = O(1)
\end{eqnarray*}
This completes the proof.
\end{proof}

The $O(1)$-approximation for uniform power claim in Thm. \ref{thm:mainresult}
 now follows from Lemma \ref{lem:maintech}. 
From Lemma \ref{lem:maintech} we can also
derive a corollary needed for the claim in Thm. \ref{thm:mainresult} comparing the performance  of the algorithm to $OPT_P$.
\begin{corollary}
Suppose at time $T$ each sender has regret at most $\epsilon$ (where $\epsilon$ is very small). Then at time $T$,
$$Q = \Omega(OPT_P/\log \Delta)$$
where $OPT_P$ is the optimum capacity under arbitrary power assignments.
\end{corollary}
This is a simple consequence of the following structural claim.
\begin{claim}
$OPT_P = O(OPT \cdot \log \Delta)$
\label{logdelta1}
\end{claim}
This claim is known for fading metrics with suitably large $\alpha$ (and thus for the plane with $\alpha > 2$), and
is due to Halld\'{o}rsson \cite{us:esa09full}. So, for fading metrics, the corollary is already implied.

We now discuss the case for general metrics. The proof of Claim \ref{logdelta1} in \cite{us:esa09full} has two
parts. First, the links are partitioned into $\log \Delta$ subsets of nearly-equilength links, ie, links whose lengths vary by no more
than 2. Then we show that for each of these sets, $OPT_P = O(OPT)$. Since there are $\log \Delta$ such sets, the 
Claim \ref{logdelta1} is proven. Now, the partitioning into $\log \Delta$ sets does not depend on the metric space used. 
It is the claim about nearly equi-length links
that must be shown to be true for general metrics. This has been recently done for
a large class of power assignments (not only uniform power) in \cite{HM10}. Since the result is yet unpublished, we
provide a self-contained proof of the claim for the case of uniform power, which is all we need.

\begin{claim}
For a set of links in any metric space where the length varies no more than a factor of 2, $OPT_P = O(OPT)$.
\end{claim}
\begin{proof}
Let us consider  an optimal $6^\alpha$-signal subset $O_1$ of $OPT_P$.
By the signal strengthening property \cite{HW09}, $|O_1| = \Omega(|OPT_P|)$.

Now let use assume that maximum capacity is achieved using power assignment $P$ that assigns to each link $v$ 
the power $P_v$.
Then, for all $v \in O_1$, $\sum_{u \in O_1} a^P_u(v) \leq \frac{1}{6^{\alpha}}$. Summing these inequalities
we get, 
\begin{eqnarray}
&& \frac{1}{6^{\alpha}} |O_1| \geq \sum_{v\in O_1} \sum_{u \in O_1} a^P_u(v) =
\sum_{v\in O_1} \sum_{u \in O_1} \frac{P_u}{P_v}\left(\frac{l_v}{d_{uv}}\right)^{\alpha} \nonumber \\
& = & \sum_{u, v \in O_1} \left(\frac{P_u}{P_v} \cdot \left(\frac{l_v}{d_{uv}}\right)^{\alpha}\right) + \left(\frac{P_v}{P_u} \cdot \left(\frac{l_u}{d_{vu}}\right)^{\alpha}\right)  \nonumber \\
& \Rightarrow & \sum_{u, v \in O_1} \left(\frac{P_u}{P_v} \cdot a_u(v) + \frac{P_v}{P_u} \cdot a_v(u)\right) \leq \frac{1}{6^{\alpha}} |O_1|
\label{nearlyeq1}
\end{eqnarray}
The calculations above are manipulations and implications of the definitions of $a^P_u(v)$ and $a_u(v)$.
We have not included $c_u$ and $c_v$ in the computations for simplicity. It suffices to notice that since
uniform power uses the maximum power, $c_u$ (and $c_v$) is smaller for $P_{\max}$ compared to its value for
$P_u$ (or $P_v$), thus the direction of the inequality works out the right way.

We claim that $a_u(v) \leq 4^{\alpha} a_v(u)$ and $a_v(u) \leq 4^{\alpha} a_u(v)$,
in other words,
$a_v(u)$ and $a_v(u)$ are comparable to each other within a constant.
First, $l_u \le 2 l_v$ and $l_v \leq 2 l_u$
by definition.  
We claim that $d_{uv} \le 2 d_{vu}$ for all $u, v$. Once we prove that, the claim between
the relation between $a_u(v)$ and $a_v(u)$ is a matter of routine calculation.
To show this, assume $P_u \leq P_v$
and we will prove the inequality in both directions. Once again ignoring $c_v$ without
loss of generality,
$\frac{P_v}{P_u}\left(\frac{l_u}{d_{vu}}\right)^{\alpha} \leq 6^{-\alpha}$. From this we get 
$d_{vu} \geq 6 l_u \geq 2(l_u + l_v)$. 
By the triangular inequality $d_{uv} \ge d_{vu} -
(l_u + l_v) \ge d_{vu}/2$ and 
$d_{vu} \ge \max(2(l_u+l_v),d_{uv}-(l_u+l_v)) \ge 2 d_{uv}/3$. 

Now continuing with the left-hand side expression of  Eqn. \ref{nearlyeq1} and using the near equality of
$a_v(u)$ and $a_u(v)$
\begin{eqnarray*}
& & \sum_{u, v \in O_1} \left(\frac{P_u}{P_v} \cdot a_u(v) + \frac{P_v}{P_u} \cdot a_v(u)\right) \leq \frac{1}{6^{\alpha}} |O_1|\\
&\Rightarrow& \sum_{u, v \in O_1} \left(\frac{P_u}{P_v} + \frac{P_v}{P_u}\right) \cdot a_u(v) \leq (4/6)^{\alpha} |O_1| \\
\end{eqnarray*}
But $(\frac{P_u}{P_v} + \frac{P_v}{P_u}) \geq 1$ for any $P_u, P_v \geq 0$. Thus,
$\sum_{u, v \in O_1} a_u(v) \leq (4/6)^{\alpha} |O_1|$, and hence $\sum_{u, v \in O_1} (a_u(v) + a_v(u)) \leq 4^{\alpha} |O_1|$. Therefore $\sum_{v \in O_1} \sum_{u \in O_1} a_u(v) \leq 4^{\alpha} |O_1|$. By a simple averaging argument,
we see that there must be a set $O_2 \subseteq O_1$ of size at least $|O_1|/2$ such that for each $v \in O_2$,
$\sum_{u \in O_2} a_u(v) \leq 2 \cdot 4^{\alpha}$. Thus, $O_2$ is a $\frac{1}{2 \cdot 4^{\alpha}}$-signal set under uniform
power. Finally, we can use signal strengthening technique once again to assert that the existence of a feasible set 
$O_3 \subseteq O_2$ of size
$\Omega(|O_2|)$. Now, $|O_3| = \Omega(|O_2|) = \Omega(|O_1|) = \Omega(|OPT_P|)$ completing the proof.
\end{proof}

\subsection{Other fixed power assignments}
In this subsection, we will discuss two other fixed power schemes that have been used in the literature.
The first is the linear power scheme where $P_v = l_v$ and the second is the mean power assignment
where $P_v = \sqrt{l_v}$. Linear power is of interest since it is power efficient in presence of noise, where
as mean power assignment has been very successful in devising centralized approximation algorithms \cite{us:esa09full}.
In both cases we will assume that the maximum possible power $P_{\max}$ is such that all links can use the 
relevant power scheme. In the same vein, in both cases, the game changes only in the detail that instead of
transmitting at full power, senders use linear (or mean) power when they decide to transmit.

\begin{corollary}
If all links use linear power, and if after time $T$ each node has small regret $\epsilon$ then
$Q = \Omega(|OPT_l|/\Delta^{\alpha})$, where $OPT_l$ is the optimum for linear power. For mean power, under the same conditions, $Q = \Omega(|OPT_m|/\Delta^{\alpha/2})$ where $OPT_m$ is the optimum for mean power.
\end{corollary}
Thus the bounds are fairly weak, yet still better than that of \cite{Dinitz2010}. The proof follows from Lemma \ref{lem:maintech} and the relation between 
affectance under linear (mean) power and affectance under uniform power. We omit the details.

Both bounds are tight. We shall provide an overview of the proof for linear power (the proof for mean power is similar). Consider a set $L = \{w\} \cup S$. Assume
 $l_w = D \gg 1$. The set $S$ is a set of $\left(\frac{D}{3}\right)^{\alpha}$ links of length $1$. The distances between the various 
 links are defined by the relation $d(s_w, s_v) = \frac{D}{2}$ for all $v \in S$. All other distances are defined
 by transitivity. Thus, for example $d_{v_1v_2} = \frac{D}{2} + \frac{D}{2} + 1 = D + 1$ (for all $v_1, v_2 \in S$).
 Note that $\Delta = D$.
Assume $\beta = 1$ and $N = 0$. It is easy to see that $S$ is feasible under linear power. Thus,
$|OPT_l| = \left(\frac{D}{3}\right)^{\alpha} = \Theta(\Delta^{\alpha})$. Now set $q_w = 1$ and 
$q_v = 0$ for all $v \in S$. We can see that each link has no-regret, because $w$ can
always transmit successfully, and while $w$ transmits no other $v \in S$ has an incentive to transmit.
Thus $Q = 1 = \Theta\left(\frac{|OPT_l|}{\Delta^{\alpha}}\right)$.

\section{Simulations}

We ran simulations to see how the distributed algorithm performs using fixed power schemes, namely uniform power where each transmitter tries to transmit with power equal to $P_{\max}$, linear power where $P_v = l_v$ and mean power where $P_v = \sqrt{l_v}$.  To implement the no-regret algorithm that each link is using, we used the Randomized Weighted Majority Algorithm of Littlestone and Warmuth \cite{Littlestone:1994:WMA:184036.184040}.  We initialized the weight of transmitting as $1$, and the weight of staying silent was also set as $1$.  In each iteration, the transmitter randomly selects whether to transmit or to stay silent, based on the weights of the corresponding actions.  The weights are updated only if the transmitter chooses to transmit.  If a transmission is successful, the weight of staying silent is multiplied by $0.5$, while if the transmission is unsuccessful, the weight of transmitting is multiplied by $0.5$.

The simulations are in the vanilla physical model, with zero ambient noise and where senders and receivers are points in the Euclidean plane.  The senders are placed uniformly at random in a square of size $100 \times 100$ in the Euclidean plane.  The receiver for each sender is placed randomly in a disc of radius $d_{\max}$ around the sender by selecting an angle uniformly at random from $[0 \; 2\pi]$ and selecting the distance between the sender and the receiver uniformly at random from $[0 \; d_{\max}]$.

The simulations with uniform power are the similar to the simulations in \cite{Dinitz2010}, so to make any comparison easier, we will usually set $\alpha = 2.1$ and $\beta = 0.5$.  As mentioned in \cite{Dinitz2010} and shown in Figures \ref{fig:increasingsize} and \ref{fig:increasingsize2}, changing these parameters does not change the trends by very much.

For comparison, we use the centralized single shot scheduling algorithm by Halld\'{o}rsson and Wattenhofer \cite{HW09}.  Their algorithm is a simple greedy algorithm, where the links are processed in a non-decreasing order of length, and each link is included in the set of active senders if the affectance of the link, caused by the current set of active links is less than or equal to a constant $c$, where $$c = \frac{1}{ \left( 2+\max \left( 2, (2^6 3 \beta \frac{\alpha-1}{\alpha-2})^{\frac{1}{\alpha}} \right) \right)^{\alpha} }$$  Even though Halld\'{o}rsson and Wattenhofer's (HW) algorithm is an $O(1)$-approximation algorithm, we realized that for the algorithm to be competitive on our random instances, the constant was too low, resulting in very small sets of active senders.  To make the HW algorithm more competitive, we improved the simulation results by using a binary search for the best constant  to determine if a link is included in the active set for each problem instance, instead of just using the fixed constant $c$.

\begin{figure}
\begin{center}
\includegraphics[width=3.6in,height=2.7in]{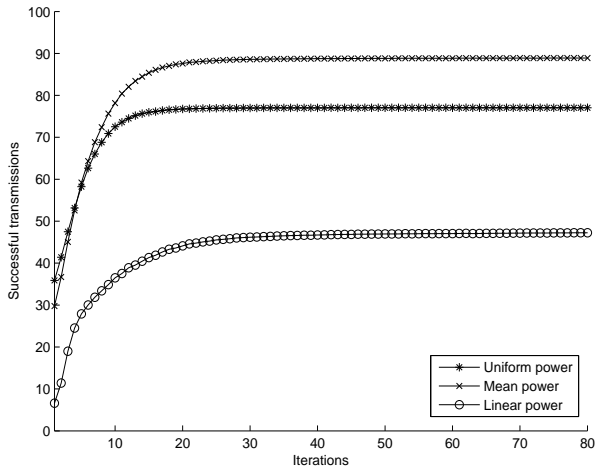}
\caption{Number of successful transmissions in each iteration for uniform, mean and linear power control schemes.  The problem instances are based on random topology with $d_{\max} = 10$, $\alpha = 2.1$ and $\beta = 0.5$.} \label{fig:convergence}
\end{center}
\end{figure}

Figure \ref{fig:convergence} shows how the distributed algorithm converges using different fixed power schemes and 200 links.  The topology is random with the maximum distance of 10 between a sender and its receiver.  Figure \ref{fig:convergence} shows the average results for the instance after solving it 10 times using the random distributed algorithm.  Regardless of the actual power control scheme, the convergence is very quick, the algorithm has usually converged to a stable solution within $30-40$ iterations, which is much less than the theoretical requirement of $O((\frac{n}{|OPT|})^2 \log n)$ iterations before the approximation guarantee can be made.  The length of the links does not seem to affect the convergence, even when $d_{\max} = 60$, the distributed algorithm reaches convergence after only $30-40$ iterations.  However, the actual number of iterations required to reach a stable solution is dependent on $n$, the number of links, but it seems to grow much more slowly than the theoretical bound indicates, although if $|OPT|$ is close to $n$, then the theoretically required number of iterations grows only as $\log n$.  When the number of links increases, the number of iterations necessary to reach convergence increases slightly.  Using $n=1000$, the distributed algorithm required about $40-50$ to reach convergence.  When the number of links is large, the mean and linear power have usually no successful links during the first iterations, while the uniform power gives successful links from the first iteration and grows steadily until it reaches convergence, similar to Figure \ref{fig:convergence}.  

\begin{figure}
\begin{center}
\includegraphics[width=3.6in,height=2.7in]{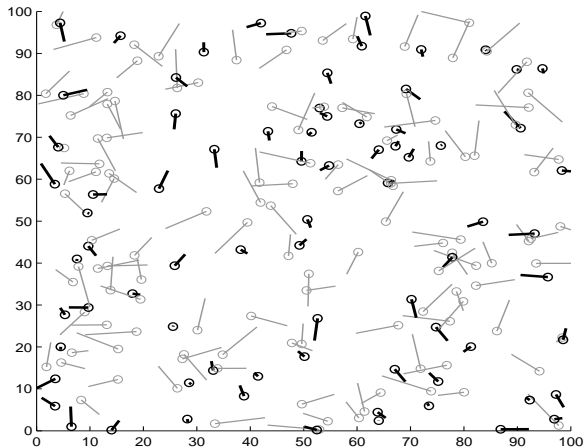}
\caption{The random instance used for comparing the results of the distributed algorithm to the optimal solution.  The transmitters are shown as small circles while the link is shown as a line segment.  The instance has $200$ links that are placed randomly using $d_{\max} = 10$.  The links that are drawn in bold denote the optimal solution using uniform power.} \label{fig:links}
\end{center}
\end{figure}

We tried to solve the instance shown in Figure \ref{fig:links} optimally using Gurobi 3.0.1 and a mixed-integer formulation of the problem.  The problem instance has $200$ links and $d_{\max}=10$, with $\alpha = 2.1$ and $\beta=0.5$.  The optimal solution with uniform power is shown in Figure \ref{fig:links} and contains $77$ active links, while the average number of successful links over 10 runs of the distributed algorithm is $73.9$, so the solution of the distributed algorithm with uniform power is very close to optimal.  The results for mean power are fairly good, the distributed algorithm finds on average $84.4$ successful links, while the incumbent solution was $97$ after running the solver for 8.5 days.  However, using linear power does not seem to work well with the distributed algorithm, the distributed algorithm only managed to find $44.9$ successful links on average for this particular instance, while the size of the optimal solution using linear power is $94$.  The HW algorithm only managed to find a set of 8 active links, while the number of active links for the HW algorithm with binary search was 51.  The capacity maximization problem in wireless networks under the SINR constraints seem to be very difficult to solve optimally, even some relatively small instances with $n = 200$ and a straightforward implementation of the mixed-integer problem could not be solved after running the solver constantly for a week on a 2.8 GHz. Intel i7 Quad-Core machine with 4 GB of memory.

\begin{figure}
\begin{center}
\includegraphics[width=3.6in,height=2.7in]{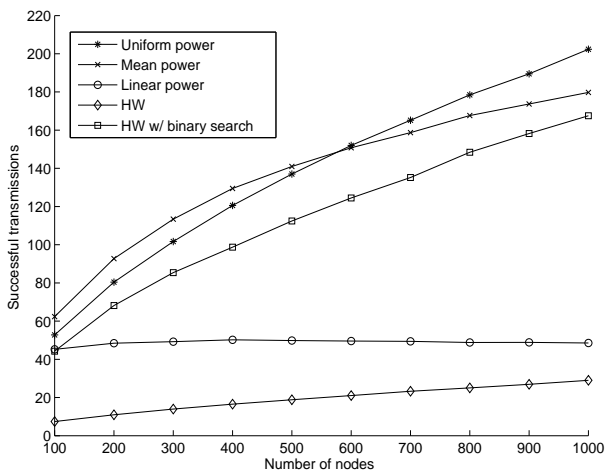}
\caption{Number of successful transmissions after 100 iterations for uniform, mean and linear power control schemes, and the results for the HW algorithm and the HW algorithm with binary search.  The problem instances are based on random topology with $d_{\max} = 10$, $\alpha = 2.1$ and $\beta = 0.5$.} \label{fig:increasingsize}
\end{center}
\end{figure}

Figure \ref{fig:increasingsize} shows the results for increasing number of links.  We created 100 instances for each value of $n$ and iterated the random distributed algorithm for 100 rounds.  We see that as $n$ gets larger, most of the algorithms do better, with the exception of the distributed algorithm using linear power.  The HW algorithm without the binary search does improve its results as $n$ gets larger, but its results are usually around $1/10$ of the best results so it is never competitive.  However, once we add the binary search to the algorithm, the performance improves greatly and the algorithm actually becomes one of the best options as the problem instances get denser.  It is interesting to note that the mean power control scheme is the best for the more sparse instances, but once the number of nodes increases above $500$, the uniform power control scheme becomes better and seems to grow more rapidly with $n$.

\begin{figure}
\begin{center}
\includegraphics[width=3.6in,height=2.7in]{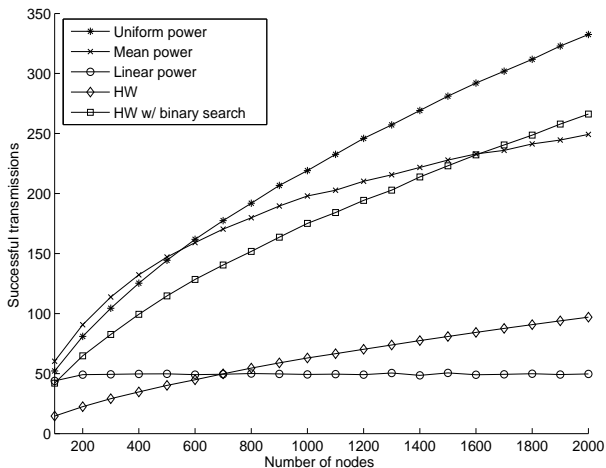}
\caption{Number of successful transmissions after 100 iterations for uniform, mean and linear power schemes, and the results for the HW algorithm and the HW algorithm with binary search.  The problem instances are based on random topology with $d_{\max} = 10$, $\alpha = 3.1$ and $\beta = 1.0$.} \label{fig:increasingsize2}
\end{center}
\end{figure}

To explore further the performance of the algorithms and how it changes as $n$ grows, as well as finding out where the intersection between mean power and HW with binary search occurs, we increased the number of links to 2000 as shown in Figure \ref{fig:increasingsize2}.  Here we use $\alpha = 3.1$ and $\beta = 1.0$, but, as mentioned earlier, the trends are very similar.  As before, $d_{\max} = 10$.  We see that, as in Figure \ref{fig:increasingsize}, the mean power scheme performs best when the number of links is below $500$, but once the number of links grows above $600$, using uniform power gives us the largest number of active links.  The number of successful links using the mean power assignment with the distributed algorithm does not grow as quickly as either using uniform power or the HW algorithms with binary search, so once the number of links grows above $1600$, the HW algorithm with binary search outperforms the distributed algorithm using the mean power scheme.  It is interesting to note that while the unmodified HW algorithm is not competitive with the more successful algorithms, the number of links still grows with $n$, whereas the performance of the distributed algorithm with linear power actually deteriorates slightly as $n$ grows larger.

\begin{figure}
\begin{center}
\includegraphics[width=3.6in,height=2.7in]{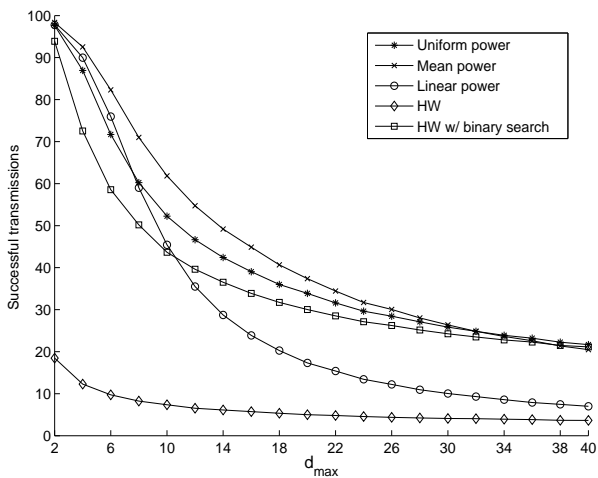}
\caption{Number of successful transmissions after 1000 iterations for uniform, mean and linear power control schemes, and the results for the HW algorithm and the HW algorithm with binary search.  The problem instances are based on random topology with $n = 100$, $\alpha = 2.1$ and $\beta = 0.5$, while $d_{\max}$ ranges from 2 to 40.} \label{fig:dmax}
\end{center}
\end{figure}

To see how the algorithms depend on the distance between the sender and the receiver, we created instances with $d_{\max}$ ranging from $2$ to $40$, where all the points are located in a square of size $100 \times 100$ as before.  Figure \ref{fig:dmax} shows how the performance of the algorithms changes as $d_{\max}$ increases.  We set $n = 100$ with $\alpha = 2.1$ and $\beta = 0.5$ as before.  The results are the average performance of the algorithms over $1000$ random instances.  The unmodified HW algorithm is still very bad and we also notice that using linear power for the distributed algorithm is only competitive with the other power control schemes when $d_{\max}$ is small, but its performance deteriorates rapidly when $d_{\max}$ grows, so it quickly becomes much worse than either the uniform or the mean power control scheme.  Figure \ref{fig:dmax} uses $n = 100$, so when we look at the results in Figures \ref{fig:increasingsize} and \ref{fig:increasingsize2}, it is not surprising to see that the mean power control scheme for the distributed algorithm obtaines the best results when $d_{\max}$ is around 10.  While the distributed algorithm using uniform and mean power and HW algorithm with binary search give very similar solutions when $d_{\max}$ is large, it seems that the HW algorithm with binary search manages to deal very well with instances where the links are likely to be long.   When we increased $d_{\max}$ above $60$, so that the links are likely to be very long and overlap each other, the HW algorithm with binary search actually starts to perform slightly better than the distributed algorithm.

\section{Conclusions}
In this paper we have improved the bounds for the game theoretic approach towards capacity maximization, and in doing so, achieved the first distributed constant factor approximation algorithm for capacity maximization for uniform power assignment.  The algorithm is a simple low-regret algorithm introduced by Andrews and Dinitz \cite{DBLP:conf/infocom/AndrewsD09} and Dinitz \cite{Dinitz2010}.  We showed that when compared to the optimum where links may use an arbitrary power assignment, the distributed algorithm achieves an $O(\log \Delta)$ approximation, where $\Delta$ is the ratio between the largest and the smallest links in the network.  The approximation factor is an exponential improvement of the existing results for distributed algorithms, and in addition, we showed that our results work for links located in any metric space.

The simulation results show that the distributed algorithm where each link runs a simple no-regret algorithm does very well in practice and in some instances, almost as good as optimal.  We show that the distributed algorithm works well in practice using uniform power for all the problem instances we tried and using the distributed algorithm with mean power also gives good results for sparse instances.  However, the simulations show that using linear power scheme does not work well with the distributed algorithm, which is consistent with a theoretical lower bound we presented.  We also give a modification to the single shot scheduling algorithm by Halld\'{o}rsson and Wattenhofer \cite{HW09}, with vastly improved practical results for random instances and show that the modified algorithm of Halld\'{o}rsson and Wattenhofer gives similar and even slightly better results than the distributed algorithm when the problem instances are very difficult.

\bibliographystyle{IEEEtran}
\bibliography{$HOME/Documents/ReadingNotes/references}

\end{document}